\documentclass[11pt,a4paper,reqno]{amsart}

\usepackage{calc,graphicx,amsfonts,amsthm,amscd,epsfig,amsmath,amssymb,enumerate,dsfont}  
\usepackage[numeric,initials,nobysame]{amsrefs}

\usepackage[showlabels,sections,floats,textmath,displaymath]{}
\usepackage[colorlinks=true, pdfstartview=FitV, linkcolor=blue,
citecolor=blue, urlcolor=blue, pagebackref=false]{hyperref} 
 \usepackage{pdfsync}
\usepackage{color}

\definecolor{refkey}{gray}{.5}
 \definecolor{labelkey}{gray}{.5}
\definecolor{light}{gray}{.9}

\reversemarginpar
\newlength\fullwidth
\numberwithin{equation}{section}

\DeclareMathSymbol{\leqslant}{\mathalpha}{AMSa}{"36} % nicer `smaller or equal'
\DeclareMathSymbol{\geqslant}{\mathalpha}{AMSa}{"3E} % nicer `larger or equal'
\DeclareMathSymbol{\eset}{\mathalpha}{AMSb}{"3F}     % nicer `emptyset'
\def\1{\ifmmode {1\hskip -3pt \rm{I}} \else {\hbox {$1\hskip -3pt \rm{I}$}}\fi}

\newtheorem{Theorem}{Theorem}[section]

\newtheorem{Lemma}[Theorem]{Lemma}

\newtheorem{remark}[Theorem]{Remark}

\newcommand{\cD}{\ensuremath{\mathcal D}}

\newcommand{\cF}{\ensuremath{\mathcal F}}

\newcommand{\bbR}{{\ensuremath{\mathbb R}} }

\newcommand{\bbZ}{{\ensuremath{\mathbb Z}} }

\def\\{\hfill\break}

\def\tthsp{\kern .083333 em}

\def\?{\mskip -10mu}
\def\indbox#1{\hbox to \parindent{\hfil\ #1\hfil} }

\newfam\msafam
\newfam\msbfam
\newfam\eufmfam
\def\hexnumber#1{%
  \ifcase#1 0\or 1\or 2\or 3\or 4\or 5\or 6\or 7\or 8\or
  9\or A\or B\or C\or D\or E\or F\fi}
\font\tenmsa=msam10
\font\sevenmsa=msam7
\font\fivemsa=msam5
\textfont\msafam=\tenmsa
\scriptfont\msafam=\sevenmsa
\scriptscriptfont\msafam=\fivemsa
\edef\msafamhexnumber{\hexnumber\msafam}%
\mathchardef\restriction"1\msafamhexnumber16
\mathchardef\ssim"0218
\mathchardef\square"0\msafamhexnumber03
\mathchardef\eqd"3\msafamhexnumber2C
\def\QED{\ifhmode\unskip\nobreak\fi\quad
  \ifmmode\square\else$\square$\fi}
\font\tenmsb=msbm10
\font\sevenmsb=msbm7
\font\fivemsb=msbm5
\textfont\msbfam=\tenmsb
\scriptfont\msbfam=\sevenmsb
\scriptscriptfont\msbfam=\fivemsb

\font\teneufm=eufm10
\font\seveneufm=eufm7
\font\fiveeufm=eufm5
\textfont\eufmfam=\teneufm
\scriptfont\eufmfam=\seveneufm
\scriptscriptfont\eufmfam=\fiveeufm

\def\({\left(}
\def\){\right)}
\let\neper=e
\let\ii=i

\let\<=\langle
\let\>=\rangle

\outer\def\nproclaim#1 [#2]#3. #4\par{\medbreak \noindent
   \talato(#2){\bf #1 \Thm[#2]#3.\enspace }%
   {\sl #4\par }\ifdim \lastskip <\medskipamount
   \removelastskip \penalty 55\medskip \fi}
\def\thmm[#1]{#1}
\def\teo[#1]{#1}
\def\sttilde#1{%
\dimen2=\fontdimen5\textfont0
\setbox0=\hbox{$\mathchar"7E$}
\setbox1=\hbox{$\scriptstyle #1$}
\dimen0=\wd0
\dimen1=\wd1
\advance\dimen1 by -\dimen0
\divide\dimen1 by 2
\vbox{\offinterlineskip%
   \moveright\dimen1 \box0 \kern - \dimen2\box1}
}
\begin{document}
\date{\today}

\title[Ensemble dependence of fluctuations]{Ensemble dependence of
  fluctuations\\
 and the canonical/micro-canonical equivalence of ensembles.}
\author[N. Cancrini]{Nicoletta Cancrini}
%\email{nicoletta.cancrini@univaq.it}
\address{Nicoletta Cancrini\\
DIIIE Universit\`a. L'Aquila, 1-67100 L'Aquila, Italy}
\email{{\tt nicoletta.cancrini@univaq.it}}

\author[S. Olla]{Stefano Olla}
%\email{stefano.olla@ceremade.dauphine.fr}
\address{Stefano Olla\\
 CEREMADE, UMR-CNRS\\
 Universit\'{e} Paris Dauphine, PSL Research University\\
 75016 Paris\\ France.}
 \email{{\tt olla@ceremade.dauphine.fr}
}

\thanks{We thanks Joel Lebowitz for pointing our attention to the microcaconical fluctuation formula
of reference \cite{LPV}, that motivated the present work. 
This paper has been partially supported by the
  European Advanced Grant {\em Macroscopic Laws and Dynamical Systems}
  (MALADY) (ERC AdG 246953) and by ANR-15-CE40-0020-01 grant LSD..
}

\begin{abstract}
We study the equivalence of microcanonical and canonical ensembles in
continuous systems, in the sense of the convergence of the
corresponding Gibbs measures. This 
is obtained by proving a local central limit theorem and a local
large deviations principle. As an application we prove a formula due to
Lebowitz-Percus-Verlet. It gives mean square fluctuations of an extensive observable,
like the kinetic energy, in a classical micro canonical ensemble at fixed energy.
\end{abstract}
\maketitle

\thispagestyle{empty}

\section{Introduction}
\label{sec:introduction}
The relation between averages of observables of a physical system with respect to different 
phase-space
% $\Gamma$- space
ensembles permits to prove what is called the {\it equivalence of ensembles}. 
That is, in the thermodynamic limit (size of the system goes to $\infty$),  
the expected value of a phase function, 
corresponding to intensive or per particle properties of the system,
 is independent of the ensemble used.
There are many different aspects and approaches  
to the equivalence of ensembles, and it will be too long to review all the literature on the subject.
For some general mathematical work we mention \cite{stroock1991microcanonical} and \cite{T}. 
We are interested here, for a system of finite $N$ particles,
 in the difference between the micro canonical average of an 
observable $A$ on a given energy shell (micro canonical manifold), 
and the canonical average of $A$ at the corresponding temperature:
\begin{equation}
  \label{eq:30}
  \Delta_N(A,u) \ =\ \left< A | u\right>_N \ - \ \left< A \right>_{N, \beta_N(u)} 
\end{equation}
where $N u$ is the value of the energy fixed in the micro canonical average, while $\beta_N(u)$ 
is the corresponding inverse temperature determined such that 
the canonical average of the energy per particle is $u$.
 We will restrict our considerations to situations 
far from phase transitions (far from thermodynamic  singularities), 
and we expect that the difference \eqref{eq:30} goes to 0 in the thermodynamic limit ($N\to\infty$). 
As the micro canonical average is just a conditional expectation 
of the canonical average for a given value of the total energy,  
this is a consequence of the concentration of the distribution of the energy per 
particle in the canonical distribution around the expected value, due to the law of large numbers. 
If $A$ is uniformly bounded in $N$, or local, 
and the micro canonical expectation $\left< A | u\right>_N$ is enough regular in $u$,
 $\Delta_N(A,u) \to 0$ is an easy consequence of a 
large deviation principle for the distribution of the energy under the canonical distribution (see
\autoref{sec:micro-canon-distr}). But here we are principally interested in extensive observables, 
like the total kinetic energy $K_N$, and their fluctuations in the micro canonical ensemble.   
In particular the micro canonical fluctuations of the total kinetic energy is greatly 
affected, and reduced, by the global constraint on the total energy and 
the asymptotic micro canonical variance, properly normalized, differs from the canonical one.  
In order to study such difference we need to compute explicitly the first order of $\Delta_N(A,u)$.

More precisely, let $\left<K_N;K_N|u\right>_N = \left<K_N^2|u\right>_N - \left<K_N|u\right>_N^2$, 
the micro canonical variance of the kinetic energy, that typically has order $N$. The canonical variance of 
$K_N$ depends only on the maxwellian distribution on the velocities and is equal to $\frac{Nn}{2\beta^{2}}$, where $n$ is the spacial dimension. 
It follows from the results contained in \autoref{sec:lebow-perc-verl} that
\begin{equation}
  \label{eq:31}
  \lim_{N\to\infty} \frac 1N \left<K_N;K_N|u\right>_N \ =\ \frac{n}{2\beta^{2}} \left(1 - \frac{n}{2 C(\beta)}\right)
\end{equation}
where the energy $u$ and inverse temperature $\beta$ are connected by the thermodynamic relation, 
and $C(\beta)$ is the heat capacity per particle, defined as $C(\beta) = \frac {d}{d{\beta^{-1}}} u(\beta)$.  
Formula \eqref{eq:31} was formally derived in \cite{LPV}, and its
rigorous derivation is the main motivation for the present article.
We actually prove \eqref{eq:31} under some regularity conditions on the micro canonical expectations, 
and in its finite $N$ version, where we also compute explicitly the next order term 
(see formula \eqref{eq:lpvk}). 
We then provide one explicit example where these regularity conditions are satisfied, but we expect that they are verified for a large class of systems.
Formula \eqref{eq:31} is actually a consequence of a more general formula \eqref{eqfin-th}, also formally deduced in  \cite{LPV},
that gives the explicit first order correction for $\Delta_N(A,u)$. 

In the proof of \eqref{eqfin-th} we use a strong form of the large deviations for the energy distribution
under the canonical measure, i.e. the asymptotic expression \eqref{funzionale} for the  density of the 
canonical probability distribution of the energy. This strong \emph{local} large deviation expression 
is proven in section \autoref{sec:local-large-devi}, 
as consequence of an Edgeworth expansion in the corresponding local central limit theorem. 
This expansion is obtained in \autoref{sec:lclt} under some condition of uniform bounds in $N$ for 
the first 4 derivatives of the free energy $f_N(\beta)$ of the canonical measure of the $N$-system.

Even though many of the arguments and results in sections 2,3 and 4 are well known in particular 
in the probabilistic literature, we decided to present this article as self contained as possible. For example 
the Edgeworth expansion argument we use in \autoref{sec:lclt} is essentially the same as used in 
Feller book \cite{F} for independent variables, but we could not find a precise reference for this statement
 for dependent continuous variables under canonical Gibbs distributions (in discrete setting see \cite{DS}, 
and general setting for dependent variables is treated in \cite{IL}).

\section{The Local Central Limit Theorem 
and its Edgeworth expansion}
\label{sec:lclt}

% In this section  we introduce the model and prove the local central limit theorem for dependent variables.
Consider $N$ particles, the momentum and coordinates given by
${\bf p}:=(p_1,\cdots, p_N)$, $p_i \in \bbR^n$ and ${\bf q}:=(q_1,\cdots,
q_N), q_i\in M$, where $M$ is a manifold of dimension $n$. 
The phase space is $\Omega^N = (\bbR^n\times M)^N$. 
Let ${\bf\bar{q}}_i=(q_1,\cdots,q_{i-1},q_{i+1},\cdots,q_N)$ be the coordinates of all the particles except that of the $i$ particle. To simplify the notation  we take $n=1$.

% The Hamiltonian of the system is given by
% \begin{equation}\label{hami2}
% H_N({\bf p},{\bf q})=\sum_{i=1}^N\left[\frac{p_i^2}{2}+ V(q_i,{\bf\bar{q}}_i)\right]
% \end{equation}
%$V$ and $U$ being regular functions.

We want to consider systems whose Hamiltonian can be written as
\begin{equation*}
H_N=\sum_{i=1}^N X_i
\end{equation*}
where \begin{equation*}
%X_1&:=\frac{p_1^2}{2}+\frac{1}{2}[V(q_1,q_2)+U(q_1)\\
X_i:=\frac{p_i^2}{2}+V(q_i,{\bf\bar{q}}_i)\quad i=1,\cdots, N\\
%X_N&:=\frac{p_N^2}{2}+\frac{1}{2}V(q_{N-1},q_N)+U(q_N)
\end{equation*}
where $V$ is a regular functions.
Define for $\beta>0$:
$$
f_N(\beta):=\frac{1}{N}\, \log\int_{\Omega^N}\, e^{-\beta H_N}d{\bf p}d{\bf q}.
$$
Notice that the integration in the $\bf p$ can always be done
explicitly and
\begin{equation*}
  f_N(\beta) = \frac 12 \log\left( 2\pi \beta^{-1} \right) + \frac{1}{N}\,
  \log\int_{M^N}\, e^{-\beta \sum_i^N V(q_i,{\bf\bar{q}}_i) } d{\bf q}.
\end{equation*}

\textbf{Assumption:}
We assume that there is an interval of values of $\beta$ such that
$f_N(\beta)$ exists, together with its first four derivatives, and
that are uniformly bounded in $N$: 
\begin{equation}
  \label{eq:11}
  \sup_N |f_N^{(j)}(\beta)| \le C_{\beta} , \quad j= 0,1,2,3,4 \
\end{equation}
with $C_{\beta}$ locally bounded in closed bounded intervals not including
$\beta=0$.  

 The canonical Gibbs measure associated to $H_N$ and temperature
 $\beta^{-1}$ is defined by
\begin{equation}\label{canmeas}
\nu_{\beta,N}(d{\bf p}\, d{\bf q})=\exp\{-\beta H_N({\bf p},{\bf
  q}) - Nf_N(\beta)\}d{\bf p}\,d{\bf q} 
\end{equation}

Defining $h_N:=H_N/N$, direct calculations give:
\begin{equation}
  \begin{split}
    f_N'(\beta) &= -\langle h_N\rangle_{\beta,N} = -u_{N}(\beta), \\
    f_N''(\beta) &= N \langle(h_N-u_{N}(\beta))^2\rangle_{\beta,N} \\
     f_N'''(\beta) &= -N^2 \langle(h_N-u_{N}(\beta))^3\rangle_{\beta,N}\\
     f_N''''(\beta) &= N^3 \langle(h_N -u_{N}(\beta))^4\rangle_{\beta,N} - 3 N   f_N''(\beta)^2.
    \label{eq:5}
  \end{split}
\end{equation}
where  we indicated $<\cdot>_{\beta,N}$ the average 
 w.r.t. the canonical measure defined in (\ref{canmeas}).

Notice that, thanks to the presence of the kinetic energy,
\begin{equation*}
  \inf_N f_N''(\beta) := \sigma_-(\beta) > \frac{1}{2\beta}.
\end{equation*}

Define the centered energy 
$$
S_N:=\sum_{j=1}^N(X_j - u_{N}(\beta))
$$
and its characteristic function
\begin{equation}\label{funcar}
\varphi_{\beta,N}(t):= \langle e^{it\,S_N}\rangle_{\beta,N}, \qquad t\in \mathbb R.
\end{equation}

By performing explicitly the integration over $\bf p$, we have
\begin{equation*}
  \varphi_{\beta,N}(t) = \left(\frac{1}{1- i t \beta^{-1}}\right)^{N/2}
  \langle e^{it\sum_j (V(q_i,{\bf\bar{q}}_i) - v_N)}\rangle_{N,\beta} 
\end{equation*}
where $Nv_N = \langle\sum_j (V(q_i,{\bf\bar{q}}_i)\rangle_{N,\beta}$. 
Consequently we have the bound:
\begin{equation}\label{cfub}
|\varphi_{\beta,N}(t)|\le \left(\frac{\beta^2}{t^2+\beta^2}\right)^\frac{N}{4},
\end{equation}
thus $|\varphi_{\beta,N}(t)|<1$ for $t\neq 0$ (i.e. is a characteristic
function of a non-lattice distribution). 
Furthermore
$|\varphi_{\beta,N}(t)|$ is integrable for $N\ge 3$, and
  by the Fourier inversion theorem (see chapter XV.3 of \cite {F})
 the probability density function of the variable $S_N$ exists for $N\ge 3$.
 Observe also that
 \begin{equation}\label{phizero}
   \begin{split}
     \varphi_N'(0) = 0, \ \varphi_N''(0) = -
     N f''_N(\beta), \ \varphi_N'''(0) = - i N f'''_N(\beta) \\
     \varphi''''_N(0) = N f''''_N(\beta) + 3 N^2  f''_N(\beta)^2.
   \end{split}
 \end{equation}

In the following we denote the normal gaussian density by
\begin{equation*}
  \phi(x) = \frac1{\sqrt{2\pi}} e^{-x/2}.
\end{equation*}
Let $\{H_j(x)\}_{j\ge 0}$ the Hermite polynomials defined by
\begin{equation}\label{Hermitep}
\frac{d^j}{dx^j}\phi(x) = (-1)^jH_j(x)\phi(x)
\end{equation}
The characteristic property of Hermite polynomials is that the
Fourier transform of $H_j(x) \phi(x)$ is given by
\begin{equation*}
  \int_{-\infty}^{+\infty} H_j(x)\phi(x) e^{itx} dx = (it)^j \hat \phi(t)
\end{equation*}
where $\hat \phi(t)=e^{-\frac{t^2}{2}}$. Recall that $H_0= 1$, $H_1(x) = x$, $H_3(x)=x^3-3x$, $H_4(x)=x^4-6x+3$ and $H_6(x)=x^6-15x^4+45x^2-15$.

We can now state the Local Central Limit Theorem we need in the rest of the article. 
% Although these methods are known to specialists of probability theory, they are not always popular in the mathematical physics literature: for these reason and to make our article self contained, we give a detailed proof of the following theorem. 
\begin{Theorem}\label{the:lclt}
Assume that $\beta$ is such that the conditions (\ref{eq:11}) are satisfied.   
Define 
$$
Y_N:=\frac{\sum_{i=1}^N(X_i-u_{N}(\beta))}{\sqrt{N f_N''(\beta)}},
$$
then the density distribution $g_{\beta,N}(x)$ of
$Y_N$ for $N\ge 3$ exists  and as $N\rightarrow\infty$ 
\begin{equation}\label{eq:lclt1}
g_{\beta,N}(x)-\phi(x)-\phi(x)\left(\frac{Q_{\beta,N}^{(3)}(x)}{\sqrt N}  + 
\frac {Q_{\beta,N}^{(4)}(x)}N  \right) 
 =o\left(\frac{1}{N}\right)\, K_N(\beta)
\end{equation} 
where
\begin{align}\label{poliP}
Q_{\beta, N}^{(3)} (x)&=\frac{f_N'''(\beta)}{3!f_N''(\beta)^\frac{3}{2}}H_3(x)%:=p_3\,H_3(x)
\\
Q_{\beta,N}^{(4)}(x)&=\frac{f_N''''(\beta)}{4! f_N''(\beta)^2}\,H_4(x)+ \frac 12\left(\frac{f_N'''(\beta)}{3!f_N''(\beta)^\frac{3}{2}}\right)^2H_6(x)\label{P4}%:=p_4\,H_4(x)+ \frac 12\, p_3^2 \,H_6(x).
\end{align}
and $K_N(\beta)$  is bounded in $N$, uniformly on bounded closed intervals of $\beta>0$. 
\end{Theorem}

\begin{proof} 

We follow the proof of theorem 2 in chapter XVI.2 of \cite{F} for independent random variables.
By (\ref{cfub}) and the Fourier inversion theorem 
the left hand side of (\ref{eq:lclt1}) exists for $N\ge 3$. To simplify the notation we do not write the dependence on $\beta$ of $f_{\beta,N}$, $\varphi_{\beta,N}$ and their derivatives.
Consider the function
\begin{equation}\label{PhiN}
 \widehat\Phi_N(t) = \varphi_N\Bigl(\frac{t}{\sqrt{N f_N''}}\Bigr)
-e^{-\frac{t^2}{2}}\left(1+P_N(\frac{it}{\sqrt{N f_N''}})\right)
\end{equation}
where $\varphi_N(t/\sqrt{Nf''_N})$  is the Fourier transform of $g_{\beta,N}$ (see (\ref{funcar}) ) and $P_N(it)$ is an appropriate polynomial in the variable $it$.
We want to show that 

\begin{equation}\label{integral}
\Delta_N= \int_{-\infty}^\infty\left|\hat\Phi_N(t)\right|\,dt =
o\left(\frac{1}{N}\right) . 
\end{equation}

Choose $\delta>0$ arbitrary but fixed. There exists  a number
$q_\delta<1$ such that
$\bigl(\frac{\beta^2}{t^2+\beta^2}\bigr)^{\frac{1}{4}}<q_\delta $ for
$|t|\ge \delta$. 
The contribution of the intervals $|t|>\delta\sqrt{f_N''N}$ to
the integral (\ref{integral}), using (\ref{cfub}),  is bounded by
\begin{equation}
 q_\delta^{N-3}\int_{-\infty}^\infty
 \left(\frac{\beta^2}{(t/\sqrt{f_N'' N})^2+\beta^2}\right)^3dt 
+\int_{|t|>\delta\sqrt{N f_N''}}e^{-\frac{t^2}{2}} |P_N(\frac{it}{\sqrt{N f_N''}})| \,dt
\end{equation}
and this tends to zero more rapidly than any power of $1/N$.

We now estimate the contribution to $\Delta_N$ from the region
${|t| \le \delta\sqrt{N f_N''(\beta)}}$.
Let us rewrite 
\begin{equation}\label{deltaN}
  \Delta_N= \int_{-\infty}^\infty
  e^{-\frac{t^2}{2}}\left|e^{\psi_N\left(t/\sqrt{f^{''}_N}\right)} - 1-P_N(\frac{it}{\sqrt{N f_N''}}) \right|\,dt 
\end{equation}
where\footnote{For a  complex number $z$ such that $|z|<1$, we define
  $\log(1+z) = \sum_n\frac{(-z)^n}{n}$.}
$$
\psi_N(t)=\log\varphi_N(t)+\frac{1}{2}Nf_N'' t^2 .
$$

The function $\psi_N(t)$ is four times differentiable and in $t=0$ its derivatives are given by

\begin{equation*}
  \psi_N'(t) = \frac{\varphi_N'(t)}{\varphi_N(t)} + Nf_N'' t, \qquad
  \psi_N'(0) = 0.
\end{equation*}
\begin{equation*}
  \psi_N''(t) = \frac{\varphi_N''(t)}{\varphi_N(t)} -
  \frac{\varphi_N'(t)^2}{\varphi_N(t)^2} + Nf_N'',
  \qquad \psi_N''(0) = 0.
\end{equation*}
\begin{equation*}
  \psi_N'''(t) = \frac{\varphi_N'''(t)}{\varphi_N(t)} -
  \frac{\varphi_N'(t)\varphi_N''(t)}{\varphi_N(t)^2}
  -\frac{2\varphi_N'(t)\varphi_N''(t)}{\varphi_N^2(t)} + \frac{2\varphi_N'(t)^3}{\varphi_N(t)^3} ,\qquad
   \psi_N'''(0) = -i N f_N'''.
\end{equation*}

\begin{equation*}
  %\begin{split}
    \psi_N''''(t) %&
    = \frac{\varphi_N''''(t)}{\varphi_N(t)} -
    \frac{3\varphi_N'(t)\varphi_N'''(t)}{\varphi_N(t)^2}
    -\frac{3\varphi_N''(t)^2}{\varphi_N^2(t)} +
    \frac{4\varphi_N'(t)^2\varphi_N''(t)}{\varphi_N(t)^3} +
    \frac{6\varphi_N'(t)^2 \varphi_N''(t)}{\varphi_N(t)^3} -  \frac{6\varphi_N'(t)^4}{\varphi_N(t)^4} \end{equation*}
    \begin{equation*}
    \psi_N''''(0) %&
    = \varphi_N''''(0) - 3 \varphi_N''(0)^2 =  N f_N'''' .
%  \end{split}
\end{equation*}
where we used relations (\ref{phizero}). Let $(it)^2\gamma_N(it)$ be the Taylor approximation for $\psi_N(t)/N$. Where $\gamma_N(it)$ is a polynomial of degree $2$ with $\gamma_N(0)=0$; it is uniquely determined by the property
\begin{equation}
\psi_N(t)-N\,(it)^2\gamma_N(it)=N o(|t|^4)
\end{equation}
and it is given by
$$
\gamma_N(it):=\frac{f_N'''}{3!}\, it+\frac{f_N''''}{4!}(it)^2
$$
We choose
$$
P_N(it):=\sum_{k=1}^2\frac{1}{k!}\,\left[N\,(it)^2\,\gamma_N(it)\right]^k
$$
then $P_N(it)$ is a polynomial in the variable $it$ with real coefficients depending on $N$ and $\beta$.
We use the inequality
$$
\left|e^\alpha-1-\sum_{k=1}^2\frac{\beta^k}{k!}\right|\le\,\left|e^\alpha-e^\beta\right|+\left|e^\beta-1-\sum_{k=1}^2\frac{\beta^k}{k!}\right|\le\, e^\gamma\left(|\alpha-\beta|+\frac{|\beta|^3}{3!}\right)
$$
with $\gamma=max\{|\alpha|,|\beta|\}$.
Furthermore we choose $\delta $ so small that for $|t|<\delta$
$$
|\psi_N(t)-N\,(it)^2\gamma_N(it)|\le \epsilon\,(f_N'')^2\,N|t|^4
$$
and
\begin{equation*}\label{delta2}
\left|\psi_N(t)\right|<\,N\,\frac 14\,f_N''t^2\quad %\text{and}
\quad\left|\gamma_N(it)\right|\le a_N |t|\le \,\frac 14\,f_N''
\end{equation*}
provided that $a_N >1+ |f_N'''|$. For  $|t|<\delta\sqrt{Nf_N''}$ the integrand in (\ref{deltaN})  can be bounded by 
\begin{equation}
e^{-\frac 14\, t^2}\,\left(\epsilon\,\frac{t^4}{N}+\frac{a_N^3}{3!}\,\left(\frac{|t|^3}{\sqrt{Nf_N''}}\right)^3\right)
\end{equation}
As $\epsilon$ is arbitrary we have that (\ref{integral}) is proved. The function $\Phi_N(t)$ defined in (\ref{PhiN}) is the Fourier transform of
\begin{equation}
g_{\beta,N}(x)-\phi(x)-\phi(x)\sum_{k=1}^8b_{Nk} H_k(x)
\end{equation}
where $b_{Nk}$ are appropriate coefficients depending on $N$ and $H_k(x)$ are the Hermite polynomials defined in (\ref{Hermitep}). If we rearrange the terms of the sum in ascending powers of $1/\sqrt{N}$ we get an expression of the form postulated in the theorem
plus terms involving powers $1/N^k$ with $k>1$ that can be dropped and obtain the result. 
\end{proof}

The same argument leads to higher order expansions, but the terms cannot be expressed by simple explicit formulas. We have the following 

\begin{Theorem}\label{the2:lclt}
Assume that $f_N''(\beta),\cdots, f_N^{(k)}(\beta)$ exist and are
uniformly bounded in $N$.  
Define $$Y_N:=\frac{\sum_{i=1}^N(X_i-u_N(\beta))}{\sqrt{N f_N''(\beta)}}$$
then the density distribution $g_{\beta,N}(x)$ of $Y_N$ for $N\ge 3$ exists  and
as $N\rightarrow\infty$ 

\begin{equation}\label{eq:lclt}
g_{\beta,N}(x)-\phi(x)-\phi(x)\sum_{j=3}^k\frac{1}{N^{\frac{1}{2}j-1}}Q^{(j)}_{\beta, N}(x)=o\left(\frac{1}{N^{\frac{1}{2}k-1}}\right)
\end{equation} 
uniformly in $x$. Here $\phi(x)$ is the standard normal density, $Q^{(j)}_{\beta, N}$
 is a real polynomial depending only on $f_N''(\beta),\cdots, f_N^{(k)}(\beta)$, 
% not on $k$,
and whose coefficients are uniformly bounded in $N$. 
\end{Theorem}
Note that Theorem \ref{the:lclt} is Theorem \ref{the2:lclt} for $k=4$ and taking $k>4$ does not improve our estimates and results.

\begin{remark}
Theorem \ref{the:lclt} is stated for continuous random variables $X_i$. It can be stated also for discrete random variables, in the same form once $|\varphi_{\beta,N}(t)|$, the characteristic function of $S_N$, is integrable. In spin systems with finite range interacting potentials, like the Ising model,  this is the case, see \cite{DT} and \cite{CM} where a Gaussian upper bound on the characteristic function is proved.
\end{remark}

\bigskip

\section{Local Large Deviations and Boltzmann formula}
\label{sec:local-large-devi}

In this section we study the energy distribution under the canonical measure.
With reasonable conditions on the interaction potential $V$,
$f_N(\beta)$ is finite for every $\beta>0$. We can extend its
definition to all $\beta \in \bbR$ denoting $f_N(\beta) = +\infty$ for
$\beta\le 0$. 

We define the Frenchel-Legendre transform of $f_N(\beta)$:
\begin{equation}\label{z*}
f_{N*}(u):=\sup_{\beta}\,\{-\beta u - f_N(\beta)\} = \sup_{\beta>0} \,\{-\beta u - f_N(\beta)\} 
\end{equation}

Let $\cD_{f_N}$, $\cD_{f_{N*}}$ the corresponding domain of
definition. % If one assume that $f_N(\beta)$ and its derivatives are
% bounded in $\beta$, these domains are equal to $\mathbb R_+$. 
For any $u\in\cD_{f_{N*}}$ there exists a unique $\beta\in\cD_{f_N}$ such that 
\begin{equation}\label{unique}
 u =-f_N'(\beta)\quad \text{and}\quad \beta=-f_{N*}'(u).
\end{equation}

Under the canonical measure (\ref{canmeas}) $h_N$ can be seen as a
normalized sum of random variables. We denote by $\cF_{N,\beta}(u)$ the
density of its probability distribution.  
For any integrable function $F\colon\bbR\rightarrow\bbR$
\begin{equation}\label{fne}
\int_{\Omega^N}F(h_N)d\nu_{\beta, N}=\int_\bbR F(u)\cF_{N,\beta}(u)du=
\int_\bbR F(u)e^{ -N[\beta u +f_N(\beta)]}W_N(u) du
\end{equation}
where 
\begin{equation}\label{WNe}
W_N(u):=\frac{d}{du}\int_{h_N\le u} d{\bf p} d{\bf q}
\end{equation}

\begin{Theorem}\label{teo1}
Let $u\in \mathcal D_{f_{N*}}$ and $\gamma=\gamma_N(u)$ 
defined by (\ref{unique}) be such that $f_N(\gamma)$ satisfies
\eqref{eq:11}. Then, for large $N$,
\begin{equation}
 W_N(u)  = e^{-Nf_{N*}(u)} 
   \sqrt{\frac{N\,f_{N*}''(u)}{2\pi}} \left( 1+\frac{Q_{\gamma(u),N}^{(4)}(0)}{N} +o\left(\frac 1N\right)K_N(\gamma(u))\right)
\end{equation}
where  $K_N(\gamma)$ and $Q_{\gamma(u),N}^{(4)}(0)$ are defined in (\ref{eq:lclt1}) and (\ref{P4}) respectively. 
%In particular $\tilde K_N(u)$ is bounded in $N$ uniformly for compact sets of values of $u$.
\end{Theorem}
\begin{proof}
Let $\omega=(\bf{p},\bf{q})\in \Omega^N$, 
${\bf  X(\omega)}=(X_1(\omega),\cdots,X_N(\omega))$,  and ${\bf
  x}=(x_1,\cdots,x_N)\in \bbR^N$. 
Consider the % probability 
positive measure
$\alpha_N(d{\bf x})$ on $\mathbb R^N$ defined, for any integrable
function $F$ on $\mathbb R^N$, by  
\begin{equation}
\int_{\Omega^N}F({\bf X}(\omega))\; d\omega =
\int_{\bbR^N}F\left( {\bf x} \right)\, \alpha_N(d{\bf x}) 
\end{equation}
so that for any $\gamma$ we have
\begin{equation}
\int_{\Omega^N}F({\bf X}(\omega))\nu_{\gamma,N}(d\omega)=
\int_{\bbR^N}F\left( {\bf x} \right)\,
e^{-\gamma\sum_{i=1}^Nx_i-Nf_N(\gamma)}\alpha_N(d{\bf x}) 
\end{equation}
For any integrable function $G: \bbR \to \bbR$ we can write
\begin{equation}
  \label{eq:2}
 \int_{\bbR^N} G\left( \frac 1N
    \sum_{j=1}^N x_j \right)\, \alpha_N(d{\bf x}) = \int_{-\infty}^{+\infty} G(s) W_N(s) ds 
\end{equation}
Take $u \in \mathcal D_{f_{N*}}$, let $\gamma= \gamma_N(u)
\in \mathcal D_{f_N}$ as in the hypotheses of the theorem. For any
integrable function $G:\bbR \to \bbR$ we have  
\begin{equation*}
  \begin{split}
    %\int_\bbR G(y) \sqrt{N} g_{\beta,N} (\sqrt N y)\; dy =
    \int_{\bbR^N} G\left(\frac 1{N\sqrt{f_N''(\gamma)}} \sum_{j=1}^N(x_j
      - u)\right) e^{-\gamma \sum_{j=1}^N x_j - Nf_N(\gamma)}
    \alpha_N(d{\bf x})\\
    = \int_\bbR G\left(\frac{s-u}{\sqrt{f_N''(\gamma)}}\right) 
      e^{-\gamma Ns- Nf_N(\gamma)} W_N(s) ds\\
      = e^{Nf_{N*}(u)} \sqrt{f_N''(\gamma)}\int_\bbR G\left(y\right) 
      e^{-\beta N \sqrt{f_N''(\gamma)} y} W_N(\sqrt{f''_N(\gamma)}y+u) dy
  \end{split}
\end{equation*}
In order to apply theorem \autoref{the:lclt} we identify
\begin{equation*}
  e^{Nf_{N*}(u)} \sqrt{f_N''(\gamma)} e^{-\gamma N \sqrt{f_N''(\beta)} y} 
  W_N(\sqrt{f''_N(\gamma)}y+u)  = \sqrt{N} g_{\gamma,N} (\sqrt N y)
\end{equation*}
so that for $y=0$
\begin{equation}\label{eq:exact}
  e^{Nf_{N*}(u)} \sqrt{f_N''(\gamma)} W_N(u)  = \sqrt{N} g_{\gamma,N} (0)
  = \sqrt{\frac{ N}{2\pi}}\left( 1+ \frac{Q_{\gamma(u),N}^{(4)}(0)}{N}+o\left(\frac{1}{N}\right)\,K_N(\gamma(u)) \right).
\end{equation}
and using $f_{N}''(\gamma)=1/f_{N*}''(u)$
\begin{equation*}
   W_N(u)  = e^{-Nf_{N*}(u)} 
   \sqrt{\frac{N\,f_{N*}''(u)}{2\pi }} \left( 1+ \frac{Q_{\gamma(u),N}^{(4)}(0)}{N}+o\left(\frac{1}{N}\right)\,K_N(\gamma(u))\right)
\end{equation*}
\end{proof}

We can resume the above result more explicitly, 
by using the bounds and the explicit form of the polynomial 
$Q_{\gamma,N}^{(4)}(0) = \frac{\gamma}{4} f^{''''}_N(\gamma)$,
\begin{equation}
  \label{eq:12}
   \left| W_N(u) e^{Nf_{N*}(u)}  \sqrt{\frac{2\pi}{N\,f_{N*}''(u)}}\, - 1 \right| \le 
   \frac{\beta_N(u) C_{\gamma_N(u)}}{4N} + o\left(\frac{1}{N}\right)\,K(\gamma_N(u)) \, .
\end{equation}
where $\gamma(u)=-f_{N*}'(u)$.

\medskip

Theorem \ref{teo1} allows to
write the probability density function in (\ref{fne}) as 
\begin{equation}\label{funzionale}
\cF_{N,\beta}(u)= e^{-NI_{N,\beta}(u)}\,  
   \sqrt{\frac{N}{2\pi}\,f_{N*}''(u)} \left( 1+\frac{Q_{\gamma(u),N}^{(4)}(0)}{N} +o\left(\frac 1N\right)K_N(\gamma(u))\right)
 %  \left( 1+ \frac{\tilde K_N(u)}{N} \right) 
   \end{equation}
where $\gamma(u)=-f_{N*}'(u)$ and 
$$
I_{N,\beta}(u):=\beta u + f_N(\beta) + f_{N*}(u) = \beta (u
-u_{N}(\beta))  -f_{N*}(u_N(\beta)) + f_{N*}(u) .
$$ 
% $u_N(\beta)=\langle
% h_N\rangle_{\beta,N}$. 
As $\beta = - f'_{N*}( u_N(\beta))$, we can thus rewrite 
\begin{equation}
  \label{eq:17}
  I_{N,\beta}(u):= f_{N*}(u)  - f_{N*}(u_N(\beta)) - f'_{N*}(u_N(\beta))  (u -u_N(\beta)) 
\end{equation}

The
functional $I_{N,\beta}(u)$ is convex, derivable and has a minimum in
$u_{\beta,N}$ where $u_{\beta,N}:=\langle
h_N\rangle_{\beta,N}$, 
$$
I_{N,\beta}'(u_{\beta,N})=0,
$$
{and}
$$ 
I''_{N,\beta}(u_{\beta,N}) = f''_{N*}(u_{\beta,N})=1/f_N''(\beta).
$$  
Equation (\ref{funzionale}) says that the sequence $h_N$
 satisfies a local large deviation principle, also called
 {\it Large Deviation Principle in the Strong Form}, see \cite{DS} where the principle is defined for discrete random variables with assumptions that are generally stronger than (\ref{eq:11}).
   
\section{Micro-Canonical distribution and equivalence of ensembles.}
\label{sec:micro-canon-distr}

We here define the equivalence of ensembles.
Given an observable $A$ on $\Omega_N$, we define the micro canonical 
average $\langle A|u\rangle_{N}$ as a conditional expectation by the classic formula:
\begin{equation}
  \label{eq:13}
  \langle A F(h_N) \rangle_{N,\beta} =  \langle \langle A| h_N\rangle
  F(h_N) \rangle_{N,\beta} =\int F(u) \langle A|u\rangle_{N}
  \cF_{N,\beta}(u) du,
\end{equation}
for any measurable function $F(u)$ on $\mathbb R$.
It is an easy exercice to see that these conditional expectations do not depend on
$\beta$. Of course \eqref{eq:13} defines the conditional expectation
only a.s. with respect to the Lebesgue measure.
But under the regularity assumptions on the interaction potential $V$,
the microcanonical surface
\begin{equation}
  \label{eq:14}
  \Sigma_N(u) = \left\{ ({\bf p}, {\bf q}) \in \Omega_N: h_N = u\right\}
\end{equation}
is regular enough such that co-area formulas (cf. \cite{EG}) can be applied and give
the existence of a regular conditional distribution on  $\Sigma_N(u)$,
defined for every value of $u$. 
We will assume in the following various conditions on the function 
$u \mapsto <A|u>_N$, 
that have to be verified in the various
applications.

By equivalence of ensembles we mean here the convergence of
\begin{equation}
  \label{eq:15}
   \langle A\rangle_{\beta,N} -   \langle A| u_{N}(\beta)\rangle_{N}
   \mathop{\longrightarrow}_{N\to\infty} 0,
\end{equation}
for a certain class of functions. We are in particular interested in the rate of convergence in (\ref{eq:15}).

For the simple case when $A$ is a bounded function
such that $<A|u>_N$ is continuous around $u=u_N(\beta)$ uniformly in
$N$,
 all we need is:
\begin{equation}
  \label{eq:16}
  \begin{split}
   & f_N(\beta) < +\infty,\qquad  \forall \beta>0\\
   & f_N(\beta)\ \text{twice differentiable in} \ \beta\\
    & \inf_N f''_N(\beta) \ge \sigma^2_- >0 .
  \end{split}
\end{equation}

By the uniform continuity of $<A|u>_N$, for any $\epsilon> 0$, there exists
$\delta_{\epsilon}> 0$ such that $ |\langle A|u\rangle_N - \langle A|u_N(\beta)\rangle_N| <
\epsilon$ if $|u-u_N(\beta)| <\delta_{\epsilon}$. Then 
\begin{equation*}
  \begin{split}
     \left| \langle A\rangle_{\beta,N} -   \langle A|
     u_{N}(\beta)\rangle_{N}\right| \le 2 \|A\|_\infty 
   \int_{|u- u_{N}(\beta)|\ge \delta_{\epsilon}} \cF_{N,\beta}(u) du +\epsilon 
  \end{split}
\end{equation*}
Let us split the large deviation term:
\begin{equation*}
   \int_{|u- u_{N}(\beta)|\ge \delta_{\epsilon}} \cF_{N,\beta}(u) du =
   \int_{u> u_{N}(\beta)+ \delta_{\epsilon}} \cF_{N,\beta}(u) du  +
   \int_{u < u_{N}(\beta)- \delta_{\epsilon}} \cF_{N,\beta}(u) du.
\end{equation*}
Let us estimate the first term of the RHS (the second term is
analogous). To shorten notation, denote $\bar u = u_{N}(\beta)+
\delta_{\epsilon}$. By exponential Chebichef inequality, for any $\lambda>0$:
\begin{equation*}
  \int_{u> \bar u} \cF_{N,\beta}(u) du \le
  e^{-N\left[\lambda\bar u  -
      f_N(\beta-\lambda) + f_N(\beta)\right]}.
\end{equation*}
Notice that
\begin{equation}
  \label{eq:18}
  \begin{split}
    I_{N,\beta} (\bar u) &= \sup_{\lambda>0} \left(\lambda \bar u -
        f_N(\beta-\lambda)\right) + f_N(\beta) \qquad u> u_N(\beta) \\
       I_{N,\beta} (\bar u) &= \sup_{\lambda<0} \left(\lambda \bar u -
        f_N(\beta-\lambda)\right) + f_N(\beta)\qquad u < u_N(\beta).
  \end{split}
\end{equation}
Consequently optimizing the estimate over $\lambda>0$ we have
\begin{equation*}
  \begin{split}
    \int_{u> \bar u} \cF_{N,\beta}(u) du &\le % \exp\left\{-N\left[
    %     \sup_{\lambda>0} \left(\lambda \bar u -  f_N(\beta-\lambda)\right) + f_N(\beta) \right]\right\} \\
    % & = 
    e^{-N I_{N,\beta}(\bar u)}
  \end{split}
\end{equation*}
and similar estimate for the term of the deviation on the other side.

Our conditions on $f''_N(\beta)$ implies the strong convexity of
$I_{N,\beta}(\bar u)$ in an interval around $u_N(\beta)$, uniform in $N$. This means exists $a>0$
such that  
$$
I_{N,\beta} (u_N(\beta) \pm \delta) \ge {a\delta^2} 
$$
It follows that
\begin{equation}\label{eq:ldub}
    \int_{|u- u_{N}(\beta)|\ge \delta_{\epsilon}} \cF_{N,\beta}(u) du
    \le 2 e^{-N a\delta_\epsilon^2}
\end{equation}
that converge exponentially to $0$ for any $\epsilon >0$. Taking
$\epsilon \to 0$ concludes the argument.

% \begin{equation*}
%    \int_{|u- u_{N}(\beta)|\le \delta_{\epsilon}} \cF_{N,\beta}(u) du \le
%    e^{-N \frac{\delta_\epsilon^2}{2\sigma^2_-}} \sqrt{\frac{2N}{\pi}}
%    \delta_\epsilon 
% \sup_{|u- u_{N}(\beta)|\ge \delta_{\epsilon}}
% \sqrt{z''_{*,N}(u)}\left(1 + \tilde K_N(u) \right)
% \end{equation*}
% which converges to $0$ since $z''_{*,N}(u)$ and $\tilde K_N(u)$ are
% uniformly bounded in compact sets. 

In the next section we will analyse closer this convergence, allowing
observables $A$ that are extensive.

\section{Lebowitz-Percus-Verlet formulas for fluctuations}
\label{sec:lebow-perc-verl}

In this section $A$ is a function on $\Omega_N$, eventually extensive, such that satisfies the
following:
\begin{itemize}
\item[(i)] $ \|A\|_{2,\beta,N}$ is finite, where 
$ \|A\|_{2,\beta,N}$ is the $L^2$ norm of $A$ with respect to the canonical measure $\nu_{\beta, N}$ defined in (\ref{canmeas})
\item[(ii)] For $j=0,1,2$ there exists $C_\beta>0$ such that
$$
\left| \frac {d^j}{du^j} \langle A| u \rangle_{N} \Big|_{u_{N}(\beta)} \right| \le\, C_\beta\, N^{j/2}  \|A\|_{2,\beta,N}, 
$$
\item[(iii)] Let $\delta_N := b\, \sqrt{\log N/N}$ for some $b>0$, then there exists $C_\beta>0$ such that
 \begin{equation}\label{deri3}
  B_{N,\beta} := \sup_{|u-u_N(\beta)|\le\delta_N}\left |\frac {d^3}{du^{3}} 
 \langle A | u\rangle_N \right| \le C_\beta\, \frac{N^{1/2}}{\log N}\, \| A\|_{2,\beta,N}.
 \end{equation} 
\end{itemize}
% \begin{enumerate}
% \item $<A|u>_N$ is three times differentiable in a neighborhood of
%   $u_N(\beta)$,
% \item Define $\mathcal I(u,u') = [u\wedge u', u\vee u']$ and
% $$
% B_N(u) :=  \frac 1{3!}  \left(\sup_{u'\in \mathcal I(u, u_N(\beta))} \frac{d^3}{du^3}
%  \langle A | u \rangle_{N}\Big|_{u'} \right) \left(u - u_{N}(\beta)\right)^3
% $$
% then assume that $\langle B_N^2\rangle_{\beta,N} < C_\beta N^p$ for some
% $C_\beta>0$ and $p<\infty$.
% \item
% Choose $\delta_N = N^{-\alpha}$ for $1/3<\alpha< 1/2$ and 
% asssume there exists a positive constant $c_\beta$ such that
% \begin{equation}\label{deri3}
% \sup_{|u-u_N(\beta)|\le\delta_N}\left |\frac {d^j}{du^{j}} \langle A | u\rangle_N
%  \right | \le c_{\beta}\,\langle |A|\rangle_{\beta,N} \qquad j= 2,3.
% \end{equation} 
% \end{enumerate}

%In the following we assume also that $f_N(\beta)^{(4)}$ is uniformly
%bounded in $N$.
\begin{Theorem}\label{lpv-th}
  Under conditions (i)-(iii) above the following formula holds
  \begin{equation}
    \begin{split}\label{eqfin-th}
      \langle A|u_{N}(\beta)\rangle_{N} =\langle A
      \rangle_{\beta,N}-\frac{1}{2N}\frac{d}{d\beta}\left[\frac{1}{f_N''(\beta)}\frac{d}{d\beta}\langle
        A \rangle_{\beta,N}\right]+o\left(\frac{1 }{N}\right) \|A\|_{2,\beta,N}.
    \end{split}
  \end{equation}
\end{Theorem}

\begin{proof}
Since expression \eqref{eqfin-th} is homogeneous in $A$, we can divide by $\|A\|_{2,\beta,N}$ 
and consider functions $A$  such that  $\|A\|_{2,\beta,N}=1$.
 We write the difference between the canonical and micro canonical
  expectations as
  \begin{equation}
    \label{eq:1}
    \langle A\rangle_{\beta,N} -   \langle A| u_{N}(\beta)\rangle_{N} = 
    \int \, \cF_{N,\beta}(u)
    \left[ \langle A| u\rangle_{N} -  \langle
      A| u_{\beta,N}\rangle_{N}\right] du 
  \end{equation}
Denote
\begin{equation*}
  G_N(u) =  \langle A| u \rangle_{N} - \langle A|u_{N}(\beta)\rangle_{N} -
    \frac d{du} \langle A| u \rangle_{N} \Big|_{u_{N}(\beta)} (u -
   u_{N}(\beta)) - \frac 12 \frac {d^2}{du^2} \langle A| u \rangle_{N}
    \Big|_{u_{N}(\beta)} (u - u_{N}(\beta))^2
\end{equation*}
Obviously $G_N(u_{N}(\beta)) =  G'_N(u_{N}(\beta)) =G''_N(u_{N}(\beta)) =0$.
We want to prove that 
\begin{equation}\label{Gu}
  \int \, \cF_{N,\beta}(u)\, G_N(u) \sim o\left(\frac{1 }{N}\right).
\end{equation}

Under conditions, $(i)-(iii)$ above, using \eqref{eq:11}, the properties of the norm and Schwarz inequality, we have that 
$ \|G_{N,\beta} \|_{\beta,N}^2 \le C_\beta'$.
For a given $\delta_N>0$, consider the bounded function
\begin{equation*}
   G_{N,\delta_N} (u) =  G_N(u) 1_{[|u-u_N(\beta)| < \delta_N]}
\end{equation*}
Then we can split the integral and,using Schwarz inequality, obtain
\begin{equation*}
%  \left| \langle G_N(h_N) \rangle_{\beta,N} \right| \ \le\
\left| \int \, \cF_{N,\beta}(u)\, G_N(u)du\right|\le
 \sqrt{C_\beta'}\, \int \, \cF_{N,\beta}(u)\, 1_{[|u-u_N(\beta)| \ge \delta_N]} du\  +\   \left|\int \, \cF_{N,\beta}(u)G_{N,\delta_N}(u)\,du \right|
\end{equation*}
By \eqref{eq:ldub}, and choosing $\delta_N = b \sqrt{\log N/N}$,
 the first term on the RHS of the above is bounded by
\begin{equation}
  \label{eq:28}
 % \langle  1_{[|u-u_N(\beta)| \ge \delta_N]} \rangle_{\beta,N} 
 \int \, \cF_{N,\beta}(u)\, 1_{[|u-u_N(\beta)| \ge \delta_N]}\, du \le 2 N^{-a b^2}
\end{equation}
and we take $b$ such that $ab^2 >1$.

For the second term, by Jensen's inequality and (\ref{funzionale}),
for any $\alpha>0$ we have
\begin{equation*}
  \begin{split}
   % \langle G_{N,\delta_N}(h_N) \rangle_{\beta,N} \le \frac 1{\alpha N}
    %\log \langle e^{\alpha N G_{N,\delta_N}(h_N)} \rangle_{\beta,N} =
   & \left|\int \, \cF_{N,\beta}(u)G_{N,\delta_N}(u)\,du \right|\le
    \frac 1{\alpha N} \log \int e^{\alpha N G_{N,\delta_N}(u)}
    \cF_{N,\beta}(u) du\\
    &= \frac 1{\alpha N} \log \left[ \int_{[|u-u_N(\beta)| < \delta_N]} 
    e^{- N\left(I_{N,\beta}(u) - \alpha G_{N,\delta}(u)\right)}
    \sqrt{\frac{N}{2\pi} f''_{N*} (u)} \left(1 + O(N^{-1})\right)\  du \ + 
   2 N^{-ab^2} \right].
  \end{split}
\end{equation*}
Since, by Taylor formula and condition $(iii)$ above, $|G_{N,\delta_N}(u)| \le B_{N,\beta}\, |u- u_N(\beta)|^3$, and  
$I_{N,\beta}(u) \ge a (u- u_N(\beta))^2$, with $a$ 
independent of $N$, we have that 
\begin{equation}
  \label{eq:29}
\begin{split}
  I_{N,\beta}(u) - \alpha G_{N,\delta}(u)& \ge (u- u_N(\beta))^2 
  \left(a - \alpha B_{N,\beta}\, |u- u_N(\beta)|\right)\\
  & \ge 
(u- u_N(\beta))^2  \left(a - \alpha B_{N,\beta} \delta_N \right)
\end{split}
\end{equation}
Choose $\alpha$ as a sequence $\alpha_N \to \infty$, for $n\to\infty$, and  such that 
$\alpha_N B_N \delta_N < a$, we have 
$$
I_{N,\beta}(u) - \alpha_N G_{N,\delta}(u) \ge 0, \qquad \text{if}\quad \ |u-u_N(\beta)| < \delta_N.
$$

Then we have:
\begin{equation*}
  \begin{split}
&N  \left|\int \, \cF_{N,\beta}(u)G_{N,\delta_N}(u)\,du \right|\\ 
&\le  \frac 1{\alpha_N} \log \left[ 2\delta_N
    \sqrt{\frac{N}{2\pi} } \sup_{|u-u_N(\beta)| <\delta_N} \sqrt{ f''_{N*} (u)} \left(1 + O(N^{-1})\right) \ + 
   2 N^{-ab^2} \right] \\
 &= \frac 1{\alpha_N} \log \left[ b\sqrt{\log N}
    \sqrt{\frac{2}{\pi} } \sup_{|u-u_N(\beta)|  <\delta_N} \sqrt{ f''_{N*} (u)} \left(1 + O(N^{-1})\right) \ + 
   2 N^{-ab^2}\right] .
  \end{split}
\end{equation*}
If $\alpha_N$ grows faster than $\log\log N$ the last term above 
tends to $0$ as $N\to\infty$ . If we choose 
$\alpha_N = \sqrt{\log N}$, we also satisfy that $\alpha_N B_N \delta_N < a$.

We can thus rewrite equation (\ref{eq:1}) as
\begin{equation}\label{f1}
  \begin{split}
    \langle A \rangle_{\beta,N} = \langle A|u_{N}(\beta)\rangle_{N} + 
     \frac{f''_N(\beta)}{2N} \frac {d^2}{du^2} \langle A| u \rangle_{N}
    \Big|_{u=u_{N}(\beta)} + o\left(\frac{1 }{N}\right)\, \|A\|_{2,\beta,N}
  \end{split}
\end{equation}
Note that for any differentiable function $g(u)$
\begin{align}\label{f21}
\frac{d}{du}\, g(u)\Big|_{u=u_N(\beta)}& 
=-\frac{1}{f_N''(\beta)}\frac{d}{d\beta}\, g(u_N(\beta))\\\label{f22}
f_N''(\beta) \frac {d^2}{du^2}\, 
g (u)\Big|_{u=u_N(\beta)}&=  \frac d{d\beta} \left[ \frac 1{f''_N(\beta)} \frac d{d\beta}\,
    g(u_N(\beta)) \right]  
\end{align}
By (\ref{f22}) we can write (\ref{f1}) as
\begin{equation}\label{rico}
\langle A|u_{N}(\beta)\rangle_{N} =\langle A
\rangle_{\beta,N}-\frac{1}{2N}\frac{d}{d\beta}\left[\frac{1}{f_N''(\beta)}\frac{d}{d\beta}\langle
  A|u_{N}(\beta)\rangle_{N}\right]+ o\left(\frac{1
  }{N}\right)\, \|A\|_{2,\beta,N}.
\end{equation}

By lemma \ref{iteration} below:
\begin{equation*}
  \frac{d}{d\beta}\left[\frac{1}{f_N''(\beta)}\frac{d}{d\beta}
    \Big(\langle A \rangle_{\beta,N} - \langle A|u_N(\beta)
    \rangle_N\Big)\right]  \sim \,o\left(\frac{1}{N}\right)\, \|A\|_{2,\beta,N} 
\end{equation*}
and \eqref{eqfin-th} follows.
% we will prove that Iterating (\ref{rico}), using (\ref{f21}) and that the four derivatives of $f_N(\beta)$  exist and are uniformly bounded in $N$ (see (\ref{eq:11}) ),
% we have
% \begin{equation}
% \begin{split}\label{eqfin}
% \langle A|u_{N}(\beta)\rangle_{N} =\langle A
% \rangle_{\beta,N}-\frac{1}{2N}\frac{d}{d\beta}\left[\frac{1}{f_N''(\beta)}\frac{d}{d\beta}\langle
%   A \rangle_{\beta,N}\right]+ <|A|>_{\beta,N}o\left(\frac{1 }{N}\right).
% \end{split}
% \end{equation}
\end{proof}

\begin{Lemma}\label{iteration} Under the conditions of Theorem \ref{lpv-th} the following relations hold  \begin{equation}
    \label{eq:19}
    \begin{split}
  \frac{d}{d\beta} \Big(\langle A \rangle_{\beta,N} - \langle A|u_N(\beta)
    \rangle_N\Big)  =\frac {f'''_{N}(\beta)}{2N} \frac {d^2}{du^2} \langle A|u
    \rangle_{N}\big|_{u=u_N(\beta)}+ \,o\left(\frac{1}{N}\right) \|A\|_{2,\beta,N} \\
  \frac{d^2}{d\beta^2}
    \Big(\langle A \rangle_{\beta,N} - \langle A|u_N(\beta)
    \rangle_N\Big) =\frac { f''''_N(\beta)}{2N}\frac {d^2}{du^2} \langle A|u
    \rangle_{N}\big|_{u=u_N(\beta)} + \, o\left(\frac{1}{N}\right) \|A\|_{2,\beta,N} \\
 \end{split}
  \end{equation}
\end{Lemma}
\begin{proof} 
Note that by (\ref{eq:1})
  \begin{equation}
    \label{eq:20}
    \begin{split}
      &\frac{d}{d\beta} \Big(\langle A \rangle_{\beta,N} - \langle
      A|u_N(\beta) \rangle_N\Big) \\
      &= -N \int \left (\langle A|u \rangle_{N} - \langle
      A|u_N(\beta) \rangle_N \right) \left( u - u_N(\beta) \right)
    \cF_{\beta,N}(u) du - \frac{d}{d\beta} \langle A|u_N(\beta) \rangle_N
    \end{split}
  \end{equation}
and, using  the definition of $G_N(u)$ above and (\ref{Gu}),
that this is equal to
\begin{equation*}
  \begin{split}
    &= -f''_N(\beta) \frac d{du} \langle A|u \rangle_{N}\big|_{u=
      u_N(\beta)} - \frac{d}{d\beta} \langle A|u_N(\beta) \rangle_N \\&+
    \frac {f'''_{N}(\beta)}{2N} \frac {d^2}{du^2} \langle A|u
    \rangle_{N}\big|_{u=u_N(\beta)} +\, o\left(\frac{1}{N}\right) \|A\|_{2,\beta,N}\\
    &=\frac {f'''_{N}(\beta)}{2N} \frac {d^2}{du^2} \langle A|u
    \rangle_{N}\big|_{u=u_N(\beta)} + \, o\left(\frac{1}{N}\right) \|A\|_{2,\beta,N}.
  \end{split}
\end{equation*}
This proves the first of \eqref{eq:19}. For the second one:
\begin{equation}
  \label{eq:21}
  \begin{split}
     & \frac{d^2}{d\beta^2}
    \Big(\langle A \rangle_{\beta,N} - \langle A|u_N(\beta)
    \rangle_N\Big) \\
    & = N^2 \int \left (\langle A|u \rangle_{N} - \langle
      A|u_N(\beta) \rangle_N \right) \left( u - u_N(\beta) \right)^2
    \cF_{\beta,N}(u) du \\
    &\quad - N  f''_N(\beta) \Big(\langle A \rangle_{\beta,N} - \langle A|u_N(\beta)
    \rangle_N\Big)
    - \frac{d^2}{d\beta^2} \langle A|u_N(\beta) \rangle_N
  \end{split}
\end{equation}
and again  using  the definition of $G_N(u)$ above and (\ref{Gu}), we have that this is equal to
\begin{equation*}
  \begin{split}
   & N^2 \frac d{du} \langle A|u \rangle_{N}\big|_{u=u_N(\beta)} 
    \int \left( u - u_N(\beta) \right)^3 \cF_{\beta,N}(u) du\\
   & + \frac{N^2}{2} \frac {d^2}{du^2} \langle A|u \rangle_{N}\big|_{u=u_N(\beta)} 
    \int \left( u - u_N(\beta) \right)^4 \cF_{\beta,N}(u) du +
   \\
    &- \frac{1}{2} \frac {d^2}{du^2} \langle A|u
    \rangle_{N}\big|_{u=u_N(\beta)} (f''_N(\beta))^2 -
    \frac{d^2}{d\beta^2} \langle A|u_N(\beta) \rangle_N +o\left(\frac{1}{N}\right) \|A\|_{2,\beta,N}\\
    & =  -f'''_N(\beta) \frac d{du} \langle A|u \rangle_{N}\big|_{u=u_N(\beta)}
   + \left(\frac 1{2N} f''''_N(\beta) + \frac 32 (f''_N(\beta))^2
   \right)  \frac {d^2}{du^2} \langle A|u
    \rangle_{N}\big|_{u=u_N(\beta)} \\
    &- \frac{1}{2}  (f''_N(\beta))^2 \frac {d^2}{du^2} \langle A|u
    \rangle_{N}\big|_{u=u_N(\beta)} -
    \frac{d^2}{d\beta^2} \langle A|u_N(\beta) \rangle_N+o\left(\frac{1}{N}\right) \|A\|_{2,\beta,N}\\
   & =  -f'''_N(\beta) \frac d{du} \langle A|u \rangle_{N}\big|_{u=u_N(\beta)}
   + \frac 1{2N} f''''_N(\beta) \frac {d^2}{du^2} \langle A|u
    \rangle_{N}\big|_{u=u_N(\beta)}\\
    & + (f''_N(\beta))^2 \frac {d^2}{du^2} \langle A|u
    \rangle_{N}\big|_{u=u_N(\beta)} -
    \frac{d^2}{d\beta^2} \langle A|u_N(\beta) \rangle_N+o\left(\frac{1}{N}\right) \|A\|_{2,\beta,N}\\
    &=  \frac{ f'''_N(\beta)}{f''_N(\beta)} \frac d{d\beta} \langle
    A|u_N(\beta) \rangle_{N} + \frac 1{2N} f''''_N(\beta) \frac {d^2}{du^2} \langle A|u
    \rangle_{N}\big|_{u=u_N(\beta)}\\
    & + f''_N(\beta) \frac{d}{d\beta} \frac 1{ f''_N(\beta)}  \frac d{d\beta} \langle
    A|u_N(\beta) \rangle_{N}  -
    \frac{d^2}{d\beta^2} \langle A|u_N(\beta) \rangle_N +o\left(\frac{1}{N}\right) \|A\|_{2,\beta,N}\\
    &=  \frac 1{2N} f''''_N(\beta) \frac {d^2}{du^2} \langle A|u
    \rangle_{N}\big|_{u=u_N(\beta)}+o\left(\frac{1}{N}\right) \|A\|_{2,\beta,N}.
  \end{split}
\end{equation*}
This proves the second of \eqref{eq:19}.

\end{proof}

Let $A$ and $B$ two functions such that they and their product satisfies the assumptions  of Theorem  \eqref{lpv-th}.
Applying formula (\ref{eqfin-th}) to $AB$ we obtain
\begin{equation*}
   \langle AB| u_{N}(\beta) \rangle_{N} = \langle AB \rangle_{N,\beta} - 
  \frac 1{2N} \frac d{d\beta} \left[\frac {1}{f''_N(\beta)} \frac
    {d\langle AB \rangle_{N,\beta}}{d\beta}\right] + o\left(\frac{1}{N}\right) \|AB\|_{2,\beta,N}\, .
\end{equation*}
while
\begin{equation*}
  \begin{split}
    \langle A| u_{N}(\beta) \rangle_{N}\langle B| u_{N}(\beta)  \rangle_{N} = &
    \langle A \rangle_{N,\beta}\langle B \rangle_{N,\beta} 
    - \frac 1{2N} \left( \langle A \rangle_{N,\beta} \frac d{d\beta}
    \left[\frac {1}{f''_N(\beta)} \frac 
      {d\langle B \rangle_{N,\beta}}{d\beta}\right] + \right.\\ 
  &\left.+\langle B
    \rangle_{N,\beta} \frac d{d\beta}
    \left[\frac {1}{f''_N(\beta)} \frac 
    {d\langle A \rangle_{N,\beta}}{d\beta}\right]\right) 
+ C_N %o\left( \frac{1}{N}\right)\,  \|A\|_{2,\beta,N}\|B\|_{2,\beta,N} 
\\ 
= &  \frac 1{N} \frac {1}{f''_N(\beta)}  \frac{d\langle A
  \rangle_{N,\beta}}{d\beta} \frac{d\langle B
  \rangle_{N,\beta}}{d\beta} -
\frac 1{2N} \frac d{d\beta} \left[\frac {1}{f''_N(\beta)} \frac
    {d(\langle A \rangle_{N,\beta}\langle B
      \rangle_{N,\beta})}{d\beta}\right]
  + C_N % o\left(\frac{1}{N}\right)\, \|A\|_{2,\beta,N}\|B\|_{2,\beta,N}
 \end{split}
\end{equation*}
where $C_N$ contains all term of smaller order and is bounded by
\begin{equation*}
  |C_N| \ \le \ o\left(\frac{1}{N}\right)\, \|A\|_{2,\beta,N}\|B\|_{2,\beta,N}.
\end{equation*}

Then defining the correlations
\begin{equation}
  \label{eq:4}
  \begin{split}
     \langle A; B | u_{N}(\beta)\rangle_{N} &:=  \langle A B |
     u_{N}(\beta)\rangle_{N} - \langle A| u_{N}(\beta)\rangle_{N}
     \langle B | u_{N}(\beta)\rangle_{N} ,\\
  \langle A; B \rangle_{\beta,N} &:=  \langle A B\rangle_{\beta,N} - \langle A\rangle_{\beta,N}
     \langle B \rangle_{\beta,N} ,
  \end{split}
\end{equation}
we get the formula for the equivalence of the correlations:
\begin{equation}
  \label{eq:6}
  \begin{split}
    \langle A; B | u_{N}(\beta)\rangle_{N} &= \langle
    A;B\rangle_{N,\beta} - \frac 1{N} \frac {1}{f''_N(\beta)}
    \frac{d\langle A \rangle_{N,\beta}}{d\beta}\frac{d\langle B
      \rangle_{N,\beta}}{d\beta} 
    - \frac 1{2N} \frac d{d\beta} \left[\frac {1}{f''_N(\beta)} \frac
    {d\langle A;B \rangle_{N,\beta}}{d\beta}\right] \\ &+o\left(\frac{ 1}{N} \right)\left(\,\|AB\|_{2,\beta,N}+
    \|A\|_{2,\beta,N}\|B\|_{2,\beta,N}\right) .
  \end{split}
\end{equation}

\begin{remark}
This formula is different than the one of reference
    \cite{LPV}.  The term with the derivative of the canonical
    correlation is in general smaller than the others. It can be
    even smaller than the error term as we will see evaluating the fluctuations of the kinetic energy below.
\end{remark}
\begin{remark}
  For extensive variables, like $A = \sum_{i=1}^N p_i^2$, typically we have 
  $\|A\|_{2,\beta,N}\sim N$, that implies that the error in  
  \eqref{eq:6} is of order $o(N)$. But in these cases the other terms are of order $N$.
\end{remark}

  %  \end{proof}

\subsection{Fluctuations of kinetic energy}
\label{sec:fluct-kinet-energy}

Consider the kinetic energy 
$$
K({\bf p})= \sum_{j=1}^N \frac{p_j^2}2.
$$ 
Then, if $n$ is the space dimension,
\begin{equation*}
  \langle K \rangle_{N,\beta} = \frac{ N n}{2\beta}, \quad \langle K^2\rangle_{N,\beta}=\frac{N(N+2)n^2}{4\beta^2}
\quad \langle K;K \rangle_{N,\beta} = \frac{Nn^2}{2\beta^2 } 
\end{equation*}
and
\begin{equation*}
  \frac {d\langle K \rangle_{N,\beta}}{d\beta}  = - \frac {Nn}{2\beta^{2}}, \quad
  \frac {d  \langle K;K \rangle_{N,\beta}}{d\beta} = - \frac{Nn^2}{\beta^{3}}
\end{equation*}
applying equation (\ref{eq:6}) we obtain
\begin{equation}
\begin{split}
 \label{eq:7}
  \langle K; K|u_N(\beta)\rangle_N - \langle K;K \rangle_{N,\beta} = &-\frac{n^2N}{4\beta^4 f''_N (\beta)}  +
  \frac 12  \frac d{d\beta} \left(\frac {n^2}{f''_N \beta^3}\right)\\ & +  o\left(\frac{ 1}{N} \right)\, \left( \|K^2\|_{2,\beta,N}+\|K\|^2_{2,\beta,N}\right)\,.
\end{split}
\end{equation}
Observe that  as $\|K\|_{2,\beta,N}\sim N/\beta$ and $ \|K^2\|_{2,\beta,N}\sim N^2/\beta^2$  
the second term in the r.h.s of (\ref{eq:7}) is smaller than the error term.
Dividing by $N$, we obtain for the variances of $K/\sqrt N$:
\begin{equation}\label{eq:lpvk}
 % \begin{split}
    \frac{1}{N}\langle K; K|u_N(\beta)\rangle_N= \frac{n}{2\beta^2 } -
    \frac {n^2}{4\beta^4 f''_N(\beta)} + o(1)\\
    = \frac{n}{2\beta^2 }\left(1 -
    \frac {n}{2C_N(\beta)}\right)+ o(1)
%  \end{split}
\end{equation}
The quantity $C_N(\beta) = \beta^2  f''_N(\beta)$ is called heat capacity (per
particle). This is in fact equal to $\frac {d}{d\beta^{-1}} u_N(\beta)$. 
Notice that \eqref{eq:lpvk} coincide, up to terms of lower order in $N$, to formula (3.7) in 
\cite{LPV}.

Notice in particular that the asymptotic canonical and microcanonical variances of 
$\frac 1{\sqrt N} K_N$ are different. 
Denoting by $V$ the total potential energy, since $K+V$ is constant under the microcanonical measure,
we have that $\langle K; K|u_N(\beta)\rangle_N = \langle V; V|u_N(\beta)\rangle_N$, so the same formula is valid for  $ \langle V; V|u_N(\beta)\rangle_N$.

It remains to prove the conditions of theorem \ref{lpv-th} are satisfied by
$\langle K_N; K_N |u\rangle_{N} $, but this in general depends on the
model considered, i.e. on the interaction between the particles.  

In \autoref{sec:local-large-devi} we have defined
\begin{equation*}
  W_N(u) = \frac{d}{du} \Omega_N(u)
\end{equation*}
where 
\begin{equation*}
  \Omega_N(u) = \int_{\mathbb R^N} d{\bf p} \int_{\mathbb R^N} d{\bf q}\ 
  \theta\left(N(u - h_N({\bf p},{\bf q}))\right)
\end{equation*}
where the Heaviside unit step function $\theta(x)$ is defined by $\theta(x)=0$ for $x<0$ 
and $\theta(x)=1$ for $x\ge 0$. 
Using the N-spherical coordinates on the momentum variables, this can be written as
\begin{equation*}
  \begin{split}
    \Omega_N(u) = S_{N-1} \int_{\mathbb R^N} d{\bf q} \int_0^\infty
    \rho^{N-1} \theta\left(Nu - \frac{\rho^2}2 - V({\bf q})\right)
    d\rho \\
    = S_{N-1} \int_{\mathbb R^N} d{\bf q}\ \theta\left(Nu - V({\bf q})\right)\int_0^{\sqrt{2(Nu - V({\bf q}))}} \rho^{N-1} d\rho \\
    = S_{N-1} \frac{2^{N/2}}{N} \int_{\mathbb R^N} d{\bf q}\ \left(Nu - V({\bf q})\right)^{\frac N2} \theta\left(Nu - V({\bf q})\right) 
 \end{split}  
\end{equation*}
where $S_{N-1}= 2\pi^{N/2}/\Gamma(N/2)$ is the surface of the $N-1$ dimensional unit sphere. 
Consequently
\begin{equation}
  \label{eq:22}
   W_N(u) = \frac {(2\pi)^{N/2}N}{\Gamma(N/2)}% S_{N-1} 2^{\frac N2 -1} N
   \int_{\mathbb R^N} d{\bf q} \left(Nu  - V({\bf q})\right)^{\frac{N}{2} - 1} \ \theta\left(Nu - V({\bf q})\right) 
\end{equation}
This formula goes back to Gibbs (\cite{gibbs}, chapter 8, (308)), 
one can prove that $W_N(u)$ is at least $\left[ \frac N2 -1\right]$ 
times differentiable see \cite{DH}. 

For any observable $A$, the micro canonical mean can be written as 
\begin{equation}\label{MCEA}
\langle A\,| u \rangle_N=
\frac{\frac{\partial}{\partial u}\int\, d{\bf p}\, d{\bf q}\, \theta(Nu-H({\bf p},{\bf q})) A({\bf p},{\bf q})}
{W_N(u)} 
% =\frac{\frac{\partial}{\partial u}\int\, d{\bf p}\, d{\bf q}\, \theta(Nu-H({\bf p},{\bf q})) A({\bf p},{\bf q})}{\frac{\partial}{\partial u}\int\, d{\bf p}\, d{\bf q}\, \theta(Nu-H({\bf p},{\bf q}))} 
\end{equation}
Using the $N$ dimensional spherical momentum coordinates as above,
one can write for the micro canonical mean of the kinetic energy  as
  \begin{equation*}
    \begin{split}
      \langle K\,|\,u\rangle_N&= W_N(u)^{-1} \left(\frac{2(2\pi)^{N/2} N}{\Gamma(N/2)} \int_{\bbR^N}\, d{\bf q}\,
        (Nu-V({\bf q}))^{\frac{N}{2}}\theta(Nu-V({\bf q})) \right)
\\ &= \frac{N^2 \Omega_N(u)}{W_N(u)} = \frac{ 2 \int_{\bbR^N}\, d{\bf q}\,
        (Nu-V({\bf q}))^{\frac{N}{2}}\theta(Nu-V({\bf q}))}
      {\int_{\bbR^N}\, d{\bf q}\, (Nu-V({\bf
          q}))^{\frac{N}{2}-1}\theta(Nu-V({\bf q}))}
    \end{split}
\end{equation*}

Of course we have the trivial bound $\langle K\,|\,u\rangle_N \le  Nu$.
Furthermore, since the micro canonical distribution is symmetric in the $\{p_j, j=1, \dots,N\}$, we have
\begin{equation}
  \label{eq:25}
   \frac 12 \langle p_j^2\,|\,u\rangle_N = \frac{ 2 \int_{\bbR^N}\, d{\bf q}\,
        (Nu-V({\bf q}))^{\frac{N}{2}}\theta(Nu-V({\bf q}))}
      {N\int_{\bbR^N}\, d{\bf q}\, (Nu-V({\bf
          q}))^{\frac{N}{2}-1}\theta(Nu-V({\bf q}))}
\end{equation}

An analogous calculation brings to 
\begin{align}\label{K2}
\langle K^2\,|\,u\rangle_N=
\frac{2^2 \int_{\bbR^N}\, d{\bf q}\, (Nu-V({\bf q}))^{\frac{N}{2} +1}\theta(Nu-V({\bf q}))}
{\int_{\bbR^N}\, d{\bf q}\, (Nu-V({\bf q}))^{\frac{N}{2}-1}\theta(Nu-V({\bf q}))} 
%\\&= N^2\frac{\int_{\bbR^N}\, d{\bf q}\, (u-V({\bf q})/N)^{\frac{N}{2}+1}\theta(u-V({\bf q})/N)}{\int_{\bbR^N}\, d{\bf q}\, (u-V({\bf q})/N)^{\frac{N}{2}-1}\theta(u-V({\bf q})/N)}\label{K2}
\end{align}
We can rewrite these expression by using the micro canonical potential energy weight:
\begin{equation}
  \label{eq:26}
  \widetilde W_N(v) := \frac{d}{dv} \int_{\bbR^N} \theta(Nv-V({\bf q})) d{\bf q}.
\end{equation}
then
 \begin{equation*}
    \begin{split}
      \langle K\,|\,u\rangle_N= N \frac{ 2\int_0^u \, (u-v)^{\frac{N}{2}}  \widetilde W_N(v) dv}
      {\int_0^u (u-v)^{\frac{N}{2}-1} \widetilde W_N(v) dv}
    \end{split}
\end{equation*}
and 
\begin{align}\label{K2-1}
\langle K^2\,|\,u\rangle_N=  4N^2  \frac{ 2\int_0^u \, (u-v)^{\frac{N}{2} +1}  \widetilde W_N(v) dv}
      {\int_0^u (u-v)^{\frac{N}{2}-1} \widetilde W_N(v) dv}
\end{align}

These formulas imply that these microcanonical averages are at least $\left[ N/2\right]$ times differentiable in $u$ and the derivatives can be explicitely computed.
% In fact
% \begin{equation}\label{eq:dermic1}
%   \begin{split}
%     \frac d{du} \langle K\,|\,u\rangle_N&=
%     N^2 - (N^2 - 2N) \frac{\langle K\,|\,u\rangle_N}{\langle K\,|\,u\rangle_{N-2}}\\
%     &= N^2\left( 1 - \frac{\langle K\,|\,u\rangle_N}{\langle K\,|\,u\rangle_{N-2}}\right) 
%     + 2 N \frac{\langle K\,|\,u\rangle_N}{\langle K\,|\,u\rangle_{N-2}}.
%   \end{split}
% \end{equation}
% \textcolor{red}{
% We need to show that 
% $N\left( 1 - \frac{\langle K\,|\,u\rangle_N}{\langle K\,|\,u\rangle_{N-2}}\right)\le C$.
% This can be deduced since we know that
% \begin{equation*}
%  0\le  \frac d{du} \langle K\,|\,u\rangle_N=  N  - \frac d{du} \langle V\,|\,u\rangle_N \le N
% \end{equation*}
% Consequently
% \begin{equation*}
%   N\left( 1 - \frac{\langle K\,|\,u\rangle_N}{\langle K\,|\,u\rangle_{N-2}}\right)\le 1
% \end{equation*}
% We also need the following lemma:
% \begin{Lemma} For some $\delta>0$, independent of $N$,
%   \begin{equation}
%     \label{eq:27}
%     \langle K\,|\,u\rangle_N \ge \delta u N
%   \end{equation}
% \end{Lemma}
% \begin{proof}
%   ???
% \end{proof}
% With the help of all these bounds and relations, 
% computing from \eqref{eq:dermic1}, we can prove
% that $\frac{d^2}{du^2}  \langle K\,|\,u\rangle_N$ satisfy \eqref{deri3}, and with a
% little more work same for  $\frac{d^3}{du^3}  \langle K\,|\,u\rangle_N$. 
% }

Starting from expression (\ref{K2}) we give a qualitative argument to understand why conditions (i)-(iii) in section \ref{sec:lebow-perc-verl} should
 be satisfied for extensive observables. 
We then present an example where most calculations can be made exactly.  
From (\ref{K2}) one can see that  dimensionally the micro canonical mean of $K^2$ behaves as 
$N^2u^2$ and that the derivatives with respect to $u$ are well defined till the order $N/2-1$. 
The third derivative of $\langle K^2\,|\,u\rangle_N$ behaves dimensionally as $N^2/u$. 
Thus, as the canonical norm $\| K\|_{2,\beta,N}^2=N(N+2)/(4\beta)$ 
and $u_N(\beta)$ does not grow in $N$, the required conditions are,
 at  least dimensionally, satisfied.
 The same reasoning can be extended to any extensive or intensive quantity 
looking directly expression (\ref{MCEA}). 

\subsection{ Exactly solvable one dimensional model}
\label{sec:oned}

We here introduce the one dimensional model system studied in \cite{DH} where conditions (\ref{deri3}) 
can be explicitly satisfied. 

Consider $N$ identical point particles confined by a one dimensional box of size $L$. The Hamiltonian is
\begin{equation}
H({\bf p},{\bf q})=\sum_{i=1}^N\frac{p_i^2}{2m}+V({\bf q})=E
\end{equation}
% The system is isolated so the energy is conserved.
The potential energy $V=V_{\rm int}+V_{\rm box}$ is determined by the interaction potential 
$$
V_{int}({\bf q})=\frac 12\sum_{i,j=1\atop i\neq j} V_{\rm pair}(|q_i-q_j|)
$$
and the box potential
$$
V_{\rm box}({\bf q})=\begin{cases}0 \quad &{\bf q}\in[0,L]^N\\
+\infty\quad & {\rm otherwise}.
\end{cases}
$$
The pair potential is given by
$$
V_{\rm pair}(r)=\begin{cases} \infty\quad & r\le d_{hc}\\
-U_0\quad & d_{hc}< r< d_{hc}+r_0\\
0\quad&r\ge d_{hc}+r_0
\end{cases}
$$
where $d_{hc}>0$ is the hard core diameter of a particle with respect to pair interactions. The pair potential above can be viewed as a simplified Lennard-Jones potential. The depth of the potential well is determined by the binding energy parameter $U_0>0$ and the interaction range by the parameter $r_0$. It is assumed  
$$
0<r_0\le d_{hc}
$$
the latter condition ensures that particles may interact with their nearest neighbors only. In order to have the volume sufficiently large for realizing the completely dissociated state, corresponding to $V=0$ it is $L>L_{\rm min}\equiv(N-1)(d_{hc}+r_0)$. The energy $E$ of the system can take values
between the ground state energy $E_0=-(N-1)U_0$ and infinity.
Following the calculations of \cite{DH} expression (\ref{K2}) for this model becomes 
\begin{equation}\label{K2mod}
\langle K^2\,|\,u\rangle_N=\frac{\sum_{k=0}^{N-1}\omega_k(Nu+kU_0)^{\frac{N}{2}+1}\theta(Nu+kU_0)}{\sum_{k=0}^{N-1}\omega_k(Nu+kU_0)^{\frac{N}{2}-1}\theta(Nu+kU_0)}
\end{equation}
where $\omega_k$ are positive coefficient depending on $N$ and $L$ see \cite{DH} for more details. Furthermore the canonical mean energy per particle 
\begin{equation*}
u_N(\beta)=\frac{1}{2\beta}-\frac{U_0}{N}\,\frac{\sum_{k=0}^{N-1}\,k\,\omega_k\,
  e^{-\beta k U_0}}{\sum_{k=0}^{N-1}\,\omega_k\, e^{-\beta k U_0}}. 
\end{equation*}
so that 
\begin{equation}\label{ucan}
\frac{1}{2\beta}-U_0\le u_N(\beta)\le \frac{1}{2\beta}
\end{equation}
Expression (\ref{K2mod}) shows that $\langle K^2\,|\, u\rangle_N$ does not vanish iff $u+\frac{N-1}{N}U_0\ge 0$ this implies $u+U_0>0$. Expression (\ref{K2mod}) is explicit but complicate. To verify that $\langle K^2\,|\,u\rangle_N$ satisfies conditions  (i)-(iii) we consider the particular case of $-1+2/N\le u< -1+3/N$ so that
$$
\langle K^2\,|\,u\rangle_N=\,\frac{\omega_{N-1}+\omega_{N-2}
\left(1-\frac1N\frac{U_0}{u+U_0}\right)^{N/2+1}}{\omega_{N-1}+\omega_{N-2}
\left(1-\frac1N\frac{U_0}{u+U_0}\right)^{N/2-1}}\,N^2(u+U_0)^2\,
$$
where we use to simplify the formulas $u+\frac{N-1}{N}U_0\sim u+U_0$ for $N$ large.
Calculating the derivatives of (\ref{K2mod})  (we omit the calculation) one can show that there exists a positive constant $A$ such that
\begin{equation}\label{K3}
\begin{split}
\langle K^2\,|\,u\rangle_N&\le \,N^2(u+U_0)^2\\
\left|\frac{d}{du}\langle K^2\,|\,u\rangle_N\right|&\le \,A\,N^2(u+U_0)\\
\left|\frac{d^2}{du^2}\langle K^2\,|\,u\rangle_N\right|&\le \,A\,N^2\left(\frac{U_0}{u+U_0}+\frac{U_0^2}{(u+U_0)^2}\right)\\
\left | \frac{d^3\langle K^2\,|\,u\rangle_N}{du^3}\right |&\le \, A\,
N^2\left[
  \frac{U_0}{(u+U_0)^2}+\frac{U_0^2}{(u+U_0)^3}+\frac{U_0^3}{(u+U_0)^4}\right] 
\end{split}
\end{equation}
Remembering that 
$$
\langle K^2\rangle_{N,\beta}=\frac{N(N+2)}{4\beta^2}
$$
by (\ref{ucan}) and (\ref{K3}) conditions (i)-(iii) of theorem \autoref{lpv-th} are satisfied.

\section{Thermodynamic limit}
\label{sec:thermodynamic-limit}

All the statements in the previous sections are for finite $N$,
under the assumption that $f_N(\beta)$ is bounded in $N$ along with
the first four derivatives. By definition $f_N(\beta)$ is analytical in $\beta$. 
Assume now that $f_N(\beta)$ converges to $z(\beta)$ which is analytical in $\beta$. Then all the derivatives of $f_N(\beta)$ converge to the derivatives of $z(\beta)$ and conditions (\ref{eq:11}) are satisfied. We thus have
\begin{equation*}
  f_N'(\beta) \to z'(\beta) = -u(\beta), \qquad  f_N''(\beta) \to
  z''(\beta) = \chi(\beta)
\end{equation*}
Usual thermodynamic notations denote $F(\beta^{-1}) = -\beta^{-1}
z(\beta)$ the \emph{free energy}, $\chi(\beta)$ \emph{heat capacity},
and $s(u) = - z^*(u) = -\lim_{N\to\infty} f_{N*}(u)$
the \emph{thermodynamic entropy}. It follows the Boltzmann formula:
\begin{equation}
  \label{eq:Boltzmann}
  s(u) = \lim_{N\to\infty} \frac 1N \log W_N(u) 
\end{equation}
Also we denote 
\begin{equation}
  \label{eq:3}
  I_\beta(u) = \lim_{N\to\infty} I_{\beta,N}(u) = \beta u - s(u) + z(\beta) 
\end{equation}
that is the rate function for the large deviations of $h_N$ is the
infinite Gibbs state defined bu DLR equations. 

In absence of phase transition, i.e. $I_\beta(u)=0$ only for $u=
z'(\beta)$, then the equivalence on ensembles follows from \eqref{eq:1}. 
Differentiability of the limit of $f_N(\beta)$ depends on the system we are considering. 
In next section we give examples where analycity  of $z(\beta)$ 
is assured at least for $\beta$  small enough.

\section{Examples}
\label{sec:example}

\subsection{Independent case}
Consider a system of $N$ noninteracting particles  in a potential. This
is the case $  V(q_i,{\bf\bar{q}}_i)] = V(q_i)$.
The Hamiltonian can be written as the sum of $N$ identical terms 
\begin{equation}\label{hamiltoniana}
H_N({\bf p},{\bf q})=\sum_{i=1}^N h(p_i,q_i)
\end{equation}
Consequently $f_N(\beta)$ does not depend on $N$ and is a smooth function of $\beta$ if $V$ is a nice reasonable potential.

\subsubsection{Independent harmonic oscillators}

Consider a system of $N$  harmonic oscillators in dimension $d$. The Hamiltonian is given by
\begin{equation}\label{ham}
H=\sum_{i=1}^{N}\left[\frac{p_i^2}{2}+\frac{q_i^2}{2}\right]
\end{equation}
To simplify notations take $n=1$. 
Explicitely we have
\begin{equation*}
  f(\beta) = \log(2\pi\beta^{-1})
\end{equation*}
and $z'(\beta) = -\beta^{-1}$, $z''(\beta) = \beta^{-2}$, so that the heat capacity here is $z''(\beta) \beta^{2} = 1$.

If we calculate the expected
value of the kinetic energy $K$ with respect to the canonical measure
at inverse temperature $\beta$ we obtain 
\begin{equation}\label{meancanK}
\langle K\rangle_\beta=\frac{N}{2\beta}
\end{equation}
The fluctuations (the {\it variance}) of $K$ are given by
\begin{equation}\label{flucanK}
\langle K; K\rangle_\beta =\frac{N}{2\beta^2}
\end{equation}
The expected value of $K$ in the with respect to the microcanonical measure is given by
\begin{equation}\label{meanmicroK}
\langle K| u\rangle_N=\frac{Nu}{2}
\end{equation}
and 
\begin{equation}
\langle K^2| u\rangle_N =\frac{N+2}{4(N+1)}\,(Nu)^2
\end{equation}
This imply that the microcanonical variance is given by
\begin{equation}\label{flumicroK}
\langle K; K | u\rangle_N =  \langle K^2| u\rangle_N -\langle K| u\rangle_N^2 %=
% \frac{3}{4}\frac{(Nu)^2}{N+1}+\frac{N-1}{4(N+1)}(Nu)^2-\frac{(Nu)^2}{4}=
= \frac{(Nu)^2}{4(N+1)}
\end{equation}
Since $\langle h_N\rangle_\beta= u_N(\beta) = \frac{1}{\beta}$, we have
\begin{equation}\label{ciao1}
\langle K; K | u_N(\beta)\rangle_N  - \langle K;K \rangle_{N,\beta} 
= \frac{N^2}{4(N+1)\beta^2} - \frac{N}{2\beta^2} 
% = \frac{N}{4\beta^2} \left(\frac N{N+1} - 2\right)
= - \frac{N}{4\beta^2} \left(1 + \frac 1{N+1} \right),
\end{equation}
that coincide with the general formula \eqref{eq:7}.
%that in the thermodynamic limit becomes
%\begin{equation}
%\langle \delta K^2\rangle_{E_\beta}-\langle \delta K^2\rangle_\beta=-\frac{N}{4\beta^2}
%\end{equation}
% Last equation permits to express the fluctuations of the kinetic energy in the micro canonical ensemble in terms of the fluctuations in the canonical ensemble which are simpler to calculate.

\subsection{Mean Field}
\label{sec:mean-field}

\begin{equation}
  \label{eq:8}
  h_N = \frac 1N \sum_1^N \frac{p_i^2}{2} + \frac 1{N^2} \sum_{i,j =1}^N
  V(q_i, q_j)
\end{equation}
Where $V$ is a symmetric reasonable potential such that $\int e^{-\beta V} dq_1
dq_2 < +\infty$ for any $\beta>0$.
One can check by direct computation, using the symmetry of the potential that
$f_N^{(j)}(\beta)$ are uniformly bounded in $N$.

\subsection{Massless surface}
\label{sec:massless-surface}
On the lattice $\bbZ^nu$:
\begin{equation}
  \label{eq:9}
   h_N = \frac 1{N^\nu} \sum_1^{N^\nu} \frac{p_i^2}{2} + \frac
   1{N^\nu} \sum_{<i,j>}^{N^\nu} V(q_i - q_j)
\end{equation}
 For $\nu = 1$ defining $r_i = q_i-q_{i-1}$, we are back to the
 independent case.

For $\nu\ge 2$, under certain conditions on $V$, there is a polynomial
decay of correlations. 
\emph{Check Spencer review}

\subsection{Real Gas}
\label{sec:real-gas}

Consider a system of $N$ particles interacting with a stable and tempered  pair potential $V\colon \bbR^d\rightarrow\bbR\cup\{\infty\}$, i.e., there exists $B\ge 0$ such that:
$$\sum_{1\le i\le j\le N}V(q_i-q_j)\ge -BN$$
for all $N$ and all $q_1,\cdots,  q_N$ and the integral
$$C(\beta)=\int_{\bbR^d}|e^{-\beta V(q)}-1|dq$$
is convergent for some $\beta>0$ (and hence for all $\beta>0$. In \cite{PT} it has been proved the validity of cluster expansion for the canonical partition function in the high temperature - low density regime. This implies that the thermodynamic free energy is analytic in $\beta$ if $\beta$ and the density are small enough. Conditions (\ref{eq:11}) are thus satisfied.

\subsection{Unbounded spin systems with finite range potential.}
\label{sec:inter-cont-spins}
We consider here the unbounded spin systems studied in \cite{BH}.
For any domain $\Lambda$ of $\bbZ^d$, with $|\Lambda |=N$, we consider the following ferromagnetic Hamiltonian on the phase space $\bbR^\Lambda$ defined as follows
$$H_N({\bf q})=\sum_{j=1}^N\left[ \phi(q_j)+\sum_{i\sim j} V( q_i, q_j)\right]=\sum_{j=1}^N X_j$$
where $i\sim j$ means that the sum is over the sites that are at distance $R>0$ from $j$. 
Here $\phi$ is a one particle phase on $\bbR$ with at least quadratic increase at infinity, $V$ is a convex function on $\bbR$ with bounded second derivative, i.e. $|V''(t)|\le C$.
As the kinetic energy term is not present to use Theorem \ref{the2:lclt} we need to prove that the characterstic function $\varphi_N(t)$ of the centered energy has modulus $|\varphi_N(t)|<1$ and $|\varphi_N(t)|$ is integrable. We have to prove an analogous of (\ref{cfub}) which assures that the probability density function of the variable $S_N$ exists. The finite range of the potential is sufficient to prove both properties.
Define a $\Lambda_R\subset\Lambda$
$$\Lambda_R=\{i\in\Lambda\,\colon\, d(i,j)>2R\}$$
and $$Y_k=\phi(q_k)+2\sum_{i\sim k}V(q_i,q_k)$$
we can write the Hamiltonian as 
$$H_N({\bf q})=\sum_{k\in\Lambda_R}Y_k+H_{\Lambda\setminus\Lambda_R}$$
where $H_{\Lambda\setminus\Lambda_R}$ depends only on the variables in $\Lambda\setminus\Lambda_R$. For any $\Lambda\subset\bbZ^d$, let $\nu_{\beta,\Lambda}$ be the canonical measure defined by the Hamiltonian defined above and indicate by $E_{\beta,\Lambda}$ the expectation value w.r.t. $\nu_{\beta,\Lambda}$. Then
\begin{align*}
\varphi_N(t)= E_{\beta,\Lambda}(e^{it\sum_{k\in\Lambda_R}Y_k+itH_{\Lambda\setminus\Lambda_R}})&=E_{\beta,\Lambda}(e^{itH_{\Lambda\setminus\Lambda_R}}\,E_{\beta,\Lambda_R}(e^{it\sum_{k\in\Lambda_R}Y_k}))\\
&=E_{\beta,\Lambda}(e^{itH_{\Lambda\setminus\Lambda_R}}\prod_{k=1}^{|\Lambda_R|}E_{\beta,k}(e^{itY_k}))
\end{align*}
where in the last equality we used independence of the $\{Y_k\}$ variables due to the finite range potential. We thus have
$$
|\varphi_N(t)|\le E_{\beta,\Lambda}(\prod_{k=1}^{|\Lambda_R|}|E_{\beta,k}(e^{itY_k})|)=E_{\beta,\Lambda}(\prod_{k=1}^{|\Lambda_R|}|\varphi_k(t)|)
$$
The variables $\{Y_k\}$ have finite probability density. This implies that their characteristic functions $\{\varphi_k(t)\}$ have modulus strictly less than one for $t\neq 0$ (see \cite{F}). Furthermore such density is in $L^2$ so that, 
 by Plancherel equality, $|\varphi_k(t)|^2$ is integrable (see \cite{F}). These two properties of $\varphi_k(t)$ assure that the modulus of $\varphi_N(t)$ is strictly less than one for $t\neq 0$ and integrable for $|\Lambda_R|$ large enough so that, by the Fourier inversion theorem, the probability density function of the centered energy exists. 

In \cite{BH} exponential decay of correlations is proven for $\beta$ small enough which implies analycity of the free energy in the thermodynamic limit.

\end{document}